\documentclass[10pt,english,aps,showkeys,notitlepage,nofootinbib,longbibliography]{revtex4-1}
\pdfoutput=1
\usepackage{color}
\usepackage[dvipsnames]{xcolor}
\usepackage{amsmath}
\usepackage{mathtools}
\usepackage{amssymb,amsthm}
\usepackage{graphicx}
\usepackage{esint}
\usepackage{babel}
\usepackage[normalem]{ulem} 
\usepackage{float}	%force figure position
\usepackage[ruled,vlined]{algorithm2e}
\usepackage{lipsum}
\usepackage[normalem]{ulem}
\usepackage[show]{ed}
\usepackage{bm}

\begin{document}

% some bindings
\def\grp{\mathcal{G}}\textsl{}
\def\ep{\mathbb E}
\def\myd{\textrm{d}}
\def\coreHD{\textsc{corehd}}
\def\WN{\textsc{weak-neighbor}}
\def\CITM{\textsc{citm}}
% thm stuff
\newtheorem{thm}{Theorem}
\newtheorem{lem}{Lemma}
\newtheorem{myconjecture}{Conjecture}
\newtheorem{mydef}{Definition}
\newtheorem{myproof}{Proof}
\def\code#1{\texttt{#1}}
\def\tpd{\the\prevdepth}

\SetKwInput{Kw}{Function}
\SetKwFunction{core}{core}

%%%%%%%%%%%%%%%% FRONT PAGE %%%%%%%%%%%%%%%%

\title{On Minimal Sets to Destroy the $k$-Core in Random Networks}

\author{Christian Schmidt$^{1}$, Henry D. Pfister$^{2}$, Lenka Zdeborov\'a$^{1}$}

\affiliation{
$^1$ Institut de Physique Th\'eorique, CEA Saclay and CNRS, 91191, Gif-sur-Yvette, France. \\
$^{2}$ Department of Electrical and Computer Engineering, Duke University, Durham, North Carolina, USA.
}

\begin{abstract}
We study the problem of finding the smallest set of nodes in a network whose removal results in an empty $k$-core; where the $k$-core is the sub-network obtained after the iterative removal of all nodes of degree smaller than $k$. This problem is also known in the literature as finding the minimal contagious set. 
The main contribution of our work is an analysis of the performance of the recently introduced \coreHD\ algorithm [Scientific Reports, {\bf6}, 37954 (2016)] on random networks taken from the configuration model via a set of deterministic differential equations. Our analyses provides upper bounds on the size of the minimal contagious set that improve over previously known bounds. Our second contribution is a new heuristic called the \WN\  algorithm that outperforms all currently known local methods in the regimes considered.
\end{abstract} 

\maketitle

\section{Introduction \label{sec:introduction}}

Threshold models are a common approach to model collective dynamical processes on networks. 
On a daily basis, we face examples such as the spreading of epidemics, opinions and decisions in social networks and biological systems. Questions of practical importance are often related to optimal policies to control such dynamics. What is the optimal strategy for vaccination (or viral marketing)? How to prevent failure propagation in electrical or financial networks?  Understanding the underlying processes, and finding fast and scalable solutions to these optimal decision problems 
are interesting scientific challenges. 
Providing mathematical insight might be relevant to the resolution of practical problems.

The contribution of this paper is related to a widely studied model for dynamics on a network: the threshold model~\cite{granovetter1978threshold}, also known as bootstrapping percolation in physics~\cite{chalupa1979bootstrap}. The network is represented by a graph $\mathcal{G}=(\mathcal{V},\mathcal{E})$, the nodes of the graph can be either in an active or inactive state. In the threshold model, a node $v \in \mathcal{V}$ changes state from inactive to active if more than $t_v-1$ of its neighbors are active. Active nodes remain active forever.
The number $t_v$ is called the threshold for node $v$. Throughout this paper we will be concerned with the case the threshold $t_v=d_v-k+1$ where $d_v$ is the degree of node $v$ and $k$ is some fixed integer. 

In graph theory, the $k$-core of a graph is defined as the largest induced subgraph of $\mathcal{G}$ with minimum degree at least $k$. This is equivalent to the set of nodes that are left after \emph{repeatedly} removing all nodes of degree smaller than~$k$.
The significance of the $k$-core of a graph for the above processes can be laid out by consideration of the complementary problem. Nodes that remain inactive, after the dynamical process has converged, must have less than $t_v$ active neighbors. Or, equivalently, they have $\kappa_v = d_v-t_v+1$ or more inactive neighbors. 
Since the dynamics are irreversible, activating a node with at least $t_v$ active neighbors may be seen as removing an inactive node with fewer than $d_v-t_v$ nodes from the graph.
In that sense, destroying the $k$-core is equivalent to the activation of the whole graph under the above dynamics with $t_v=d_v-k+1$. In other words, the set of nodes that leads to the destruction of the $k$-core is the same set of nodes that lead, once activated, to the activation of the whole graph. 
Following the literature~\cite{Reichman2012,Coja-Oghlan14,GuggiolaSemerjian2015} we call the smallest such set the {\it minimal contagious sets} of $\mathcal{G}$. In a general graph this decision problem is known to be NP-hard~\cite{dreyer2009irreversible}. 

Most of the existing studies of the minimal contagious set problem are algorithmic works in which algorithms are proposed and heuristically tested against other algorithms on real and synthetic networks; see e.g.~\cite{altarelli2013optimizing,altarelli2013large,Braunstein2016,ZdeborovaZhangZhou2016,MakseMorone2017} for some recent algorithmic development.  Theoretical analysis commonly brings deeper understanding of a given problem. In the case of minimal contagious sets, such analytic work focused on random graphs, sometimes restricted to random regular graphs, in the limit of large graphs. On random graphs the minimal contagious set problem undergoes a threshold phenomenon: in the limit of large graph size the fraction of nodes belonging to the minimal contagious set is with high probability concentrated around a critical threshold value.

To review briefly the theoretical works most related to our contribution we start with the special case of $k=2$ that has been studied more thoroughly than $k\ge 3$. The choice of $k=2$ leads to the removal of the $2$-core, i.e. removal of all cycles, and is therefore referred to as the decycling problems. This case is also known as the feedback vertex set, one of the 21 problems first shown to be NP-complete~\cite{karp1972reducibility}. On random graphs the decycling problem is very closely related to network dismantling, i.e.\ removal of the giant component~\cite{Braunstein2016}. 
A series of rigorous works analyzes algorithms that are leading to a the best known bounds on the size of the decycling set in random regular graphs~\cite{BauWormaldZhou02,HoppenWormald08}. 
Another series of works deals with the non-rigorous cavity method and estimated values for the decycling number that are expected to be exact or extremely close to exact~\cite{zhou2013spin,GuggiolaSemerjian2015,Braunstein2016}. 
The cavity method predictions also lead to specific message passing algorithms~\cite{altarelli2013optimizing,altarelli2013large,Braunstein2016}. While there is some hope that near future will bring rigorous establishment on the values of the  thresholds predicted by the cavity method, along the lines of the recent impressive progress related to the $K$-SAT problem~\cite{ding2015proof}, it is much further of reach to analyze rigorously the specific type of message passing algorithms that are used in~\cite{altarelli2013optimizing,altarelli2013large,Braunstein2016}. A very well performing algorithm for decycling and dismantling has been recently introduced in~\cite{ZdeborovaZhangZhou2016}. One of the main contributions of the present work is to analyze exactly the performance of this algorithm thus leading to upper bounds that are improving those of~\cite{BauWormaldZhou02} and partly closing the gap between the best know algorithmic upper bounds and the expected exact thresholds~\cite{GuggiolaSemerjian2015,Braunstein2016}.

The case of contagious sets with $k\ge 3$ is studied less broadly, but the state-of-the-art is similar to the decycling problem. Rigorous upper bounds stem from analysis of greedy algorithms~\cite{Coja-Oghlan14}. The problem has been studied very thoroughly via the cavity method on random regular graphs in~\cite{GuggiolaSemerjian2015}, result of that paper are expected to be exact or very close to exact.

\subsection{Summary of our contribution \label{subsec:contributions}}

This work is inspired by the very simple decycling algorithm \coreHD\  proposed in~\cite{ZdeborovaZhangZhou2016}. In numerical experiments the authors of~\cite{ZdeborovaZhangZhou2016} found that \coreHD\ is close in performance to the message passing of~\cite{Braunstein2016}, which is so far the best performing algorithm for random graphs. The \coreHD\  algorithm of~\cite{ZdeborovaZhangZhou2016} proceeds by iterating the following two steps: first the $2$-core of the graph is computed and second a largest-degree-node is removed from the $2$-core. It was already anticipated in~\cite{ZdeborovaZhangZhou2016}  that this algorithm can be easily extended to the threshold dynamics where we aim to remove the $k$-core by simply replacing the $2$ with a $k$.

In this work, we observe that the dynamics of the \coreHD\  algorithm is amenable to rigorous analysis and performance characterization for random graphs drawn from the configuration model with any bounded degree distribution.  
We show that it is possible to find a deterministic approximation of the macroscopic dynamics of the \coreHD\  algorithm. Our results are based on the mathematical analysis of peeling algorithms for graphs by Wormald in~\cite{Wormald1995}.  Our particular treatment is inspired by the tutorial treatment in~\cite{Pfister14analysis} of the decoding analysis in~\cite{Luby-it01}.
Clearly the algorithm cannot be expected to perform optimally as it removes nodes one by one, rather than globally. 
However, our work narrows considerably the gap between the existing rigorous upper bounds~\cite{BauWormaldZhou02, Coja-Oghlan14} and the expected optimal results~\cite{GuggiolaSemerjian2015}. 
Note also that the fully solvable relaxed version -- where first \emph{all} nodes of the largest degree are removed from the core before the core is re-evaluated -- yields improvements.

Our analysis applies not only to random regular graphs, but also to random graphs from the configuration model defined by a degree distribution.
The basic theory requires that the degree distribution is bounded (i.e. all the degrees are smaller than some large constant independent of the size of the graph). 
But, the most commonly used Erd\H{o}s-R\'enyi random graphs have a Poisson degree distribution whose maximum degree grows slowly.
Fortunately, one can add all the nodes with degree larger than a large constant to the contagious set.
If the fraction of thus removed edges is small enough, then the asymptotic size of the contagious set is not affected and the same result holds.
Results are presented primarily for random regular graphs in order to compare with the existing results. 

The following results are presented.

\paragraph{Exact analysis of the \coreHD\  algorithm.}
We show that the \coreHD\  algorithm (generalized to $k$-core removal) translates to a random process on the degree distribution of the graph $\mathcal{G}$. We track this random process by derivation of the associated continuous limit. This is done by separating the random process into two deterministic processes and absorbing all randomness into the running time of one of them.
We derive a condition for the running time in terms of the state of the system and this reduces the dynamics to a set of coupled non-linear ordinary differential equations (ODEs) describing the behaviour of the algorithm on a random graph, eqs.~(\ref{eq:dgl_final1}--\ref{eq:dgl_final3}).

\paragraph{New upper bounds on the size of the minimal contagious sets in random graphs.}
The stopping time of the before-mentioned ODEs is related to the number of nodes that were removed from the $k$-\emph{core} during the process. Thus providing upper bounds on the expected minimal size of the contagious set of $\mathcal{G}$. A numerical evaluation shows that the bounds improve the best currently known ones~\cite{BauWormaldZhou02} and narrow the gap to the anticipated exact size of the minimal contagious set from~\cite{GuggiolaSemerjian2015}, see e.g. table \ref{tbl:bounds} for the decycling, $k=2$, problem.

\paragraph{Improved heuristic algorithm.}
Based on intuition we gained analyzing the \coreHD\  algorithm we propose it's extension that further improves the performance. In this new algorithm, instead of first removing high degree nodes, we first remove nodes according to the decision rule $\arg\max \left[ d_i -\frac{1}{d_i}\sum_{j\in \partial i} d_j \right]$. On graphs with bounded degree this algorithm has $O(N)$ running time, where $N$ is the number of nodes in the graph. In experiments we verify that this \WN\  heuristics improves over \coreHD\  and other recently introduced algorithms such as the \CITM\ of~\cite{MakseMorone2017}.

The paper is organized in two main parts. The first part in section \ref{sec:analysis} is devoted to the analysis of the generalized \coreHD\  algorithm and comparison of the resulting upper bounds with existing results. In the second part in section \ref{sec:improving} we introduce the new algorithm called the \WN\  
(that we do not study analytically) and close with some numerical experiments and comparison with other local algorithms.

\section{The analysis of the \coreHD\  algorithm \label{sec:analysis}}

In Algorithm \ref{alg:hd} we outline the \coreHD\  algorithm of~\cite{ZdeborovaZhangZhou2016}, generalized from $k=2$ to generic $k$. The algorithm provides us with a {\it contagious set} of nodes $\mathcal{D}$ such that after their removal the resulting graph has an empty $k$-core.
Consequently the size of  $\mathcal{D}$ provides an upper bound on the size of the minimal contagious set. 
In terms of Algorithm~\ref{alg:hd}, our aim is to show that the size of $\mathcal{D}$ per node has a well defined limit, and to compute this limit. 
%If asymptotically the fluctuations around the expectation are vanishing the size of $\mathcal{D}$ is almost surely equal to its average.

\SetKwData{KK}{$k$}
\SetKwData{CC}{$\mathcal{C}$}
\SetKwData{GG}{$\mathcal{G}$}

\begin{algorithm}[H] 
	\SetAlgoNoLine
	\LinesNumbered
		\KwData{$\mathcal{G}(\mathcal{V},\mathcal{E})$}
		\KwResult{A set of nodes $\mathcal{D}$ whose removal makes the $k$-core vanish}
		\Kw{\core(\KK,\GG) returns the $k$-core of the graph $\mathcal{G}$}
		Initialize: $\mathcal{C} \gets \core(\KK,\GG)$ , $\mathcal{D} = \{ \ \}$ \\
		\While{$\left|\mathcal{C}\right|>0$}
		{
			$\mathcal{M} \gets \{i \in \mathcal{V}_{\mathcal{C}} \mid i = \arg\max \left[ d_i \right]\}$ 
			\tcp*{$\mathcal{V}_{\mathcal{C}}$ is the set of nodes in $\mathcal{C}$}
			$r \gets \textrm{uniform}(\mathcal{M})$ \;
			$\mathcal{C} \gets \mathcal{C}\backslash r$ \;
			$\mathcal{D} \gets \mathcal{D} \cup \{r\}$ \;
			$\mathcal{C} \gets \core(\KK,\CC) $ \;
		}
		\caption{Generalized \coreHD\  algorithm. }
		\label{alg:hd}
	\end{algorithm}

With a proper book-keeping for the set $\mathcal{M}$ and dynamic updating of the $k$-core the running time of \coreHD\  on graphs with bounded degree is $O(N)$. The algorithm can be implemented such that in each iteration exactly one node is removed: if a node of degree smaller than $k$ is present, it is removed, else a node of highest degree is removed. Thus running the algorithm reduces to keeping track of the degree of each node. If the largest degree is $O(1)$ this can be done in $O(1)$ steps by binning all the nodes of equal degree. If a node is removed only the degrees of all its $O(1)$ neighbours must be moved to the new adequate bins. To review and test our implementation of the \coreHD\ algorithm see the open depository \cite{demo_git_Christian}.

\subsection{Reduction into a random process on the degree distribution \label{sec:reduction}}

In the next several sections, we derive closed-form deterministic equations for the macroscopic behavior of the \coreHD\  algorithm in the limit of large random graphs taken from the configuration model.
This is possible because, when the \coreHD\  algorithm is applied to a random graph from the configuration model (parameterized by its degree distribution), the result (conditioned on the new degree distribution) is also distributed according to the configuration model.
Thus, one can analyze the \coreHD\  algorithm by tracking the evolution of the degree distribution.

In particular, the behaviour of the \coreHD\  procedure averaged over the graph $\mathcal{G}$ can be described explicitly in terms of the following process involving colored balls in an urn~\cite{Pfister14analysis}. 
At time step $n$ there will be $N_n$ balls in the urn, each of which carries a color $d_i$, with $d_i \in \{1,\dots,d \} $ and $d$ being the maximum degree in the graph at the corresponding time step. At any time step $n$ there are $v_q(n)$ balls of color $q$. The colors of the balls are initialized in such a way that at time $n=0$, the number of balls, $N_0$, is the size of the $k$-core of the original graph $\mathcal{G}$ and the initial values of their colors are chosen from the degree distribution of the $k$-core (relation of this to the original degree distribution is clarified later).

In a first step, called ``removal''  (line 3-6 in Alg.~\ref{alg:hd}), one ball is drawn among the $v_d(n)$ balls of maximum degree (color $d$) uniformly at random. Next $d$ balls $\{i_1,i_2,\dots,i_d\}$ are drawn with colors following the excess degree distribution of the graph. The excess degree distribution $P(q) = q Q(q) / c$ of a graph of degree distribution $Q(q)$ and average degree $c$ gives the probability that an outgoing edge from one node is incoming to another node of degree $q$.  To conclude the first step, each of the $d$ balls is replaced by a ball of color $d_{i_j}-1$ (we assume that there are no double edges).
In a second step, called ``trimming''  (line 7 in Alg.~\ref{alg:hd}), we compute the $k$-core of the current graph. In the urn-model this is equivalent to repeatedly applying the following procedure until $v_{i}=0 \ \forall \ i<k$: draw a ball of color $q \in \{ 1,\dots,k-1 \}$ and relabel $q$ other balls chosen according to the excess degree distribution. Thus, we obtain a random process that depends purely on the degree distribution.
Note that in this process we used the fact that the graph was randomly drawn from the configurations model with degree distribution $Q(q)$.

One difficulty, when analyzing the above process, is to chose the right observables. In the previous paragraph the nodes were used as observables. However, equally, one might consider the process in terms of the edges of the graph. As outlined in the previous paragraph, it is important to keep track of the excess degree distribution. Henceforth we will be working with the \emph{half-edges} to simplify the analysis.

To reinterpret the above process in terms of half-edges let $h_i$ be the total number of half-edges that are connected to nodes of degree $i$ at the current iteration. Furthermore, we distinguish nodes of degree smaller than $k$ from all the others. To do so we \emph{adapt our index notation in what follows} and identify $h_{k_<} \coloneqq \sum_{i<k} h_i$ and denote the sum over the entries of a vector $\bm{h}(n) = (h_{k_<}(n),h_k(n),\dots,h_d(n))^T$ as $\left| \bm{h}(n) \right|$. Finally let us also define the unit vectors $\mathbf{e}_{k_<} = (1,0,\dots,0)^T$, $\mathbf{e}_{k} = (0,1,0,\dots,0)^T$, \dots, $\mathbf{e}_d=(0,\dots,0,1) ^T$. Each ball now represents a half-edge and its color is according to the degree of the node that this half-edge is connected to.

The two steps (trimming and removal) can be described in terms of half-edges as follows. 
We start with the ``removal'' step (line 3-6 in Alg.~\ref{alg:hd}). It can be recast in the following rule
	\begin{align}
%	\tag{$\alpha$}
	\begin{split}
	%1.i)  & \hspace{0.5cm}  \bm{h} \gets \bm{h} - d\cdot \mathbf{e}_d 
	& (i) \hspace{0.31cm} \bm{h} \gets \bm{h} - d\cdot \mathbf{e}_d 
	\\
	%1.ii) & \hspace{0.5cm}  \text{Repeat $d$ times: } \bm{h} \gets \bm{h} + \mathbf{A} \bm{z}
	& (ii) \hspace{0.2cm} \text{Repeat $d$ times: } \bm{h} \gets \bm{h} + \mathbf{A} \bm{z}
	\end{split}
	\label{eq:iterate_hd_1}
	\end{align}
where the vector $\bm{z} \in {\mathbb R}^{d-k+2} $ is a random vector that has zeros everywhere except in one of the $d-k+2$ directions, in which it carries a one. The probability that $\bm{z}$ is pointing in direction at iteration $n$ is given by the excess degree distribution $h_q(n)  / \left| \bm{h}(n) \right|$ for $q = \, k_<, k, k+1,  \dots, d$. 
When a node of degree $d$ is removed from the graph, together with its half-edges, the remaining cavity leaves behind some dangling half-edges that are pruned away in step $(ii)$ using the following ``relabelling'' matrix
	\begin{equation}
	\mathbf{A} = 
		\begin{pmatrix}
		-1 & k-1 & & &  \\
		  & -k & k & &  \\
		 & & \ddots & \ddots & \\
		 & & & -(d-1) & d-1 \\
		 & & & & -d
		\end{pmatrix}
		\in \mathbb{R}^{(d-k+2)\times(d-k+2)}
		\, .
	\label{eq:relabeling_matrix}
	\end{equation}

Analogously, the ``trimming'' step (line 7 in Alg.~\ref{alg:hd}) can be cast in the following update rule where step $(i)$ removes a \emph{single half-edge} of degree $<k$ and subsequently step $(ii)$ trims away the dangling cavity half-edge
	\begin{align}
%	\tag{$\beta$}
	\begin{split}
	%2) & \hspace{0.5cm} \text{while} \ h_{<k} > 0 \  \text{iterate} 
	& \text{while} \ h_{k_<} > 0 \  \text{iterate:} 
	\\
	%& \hspace{0.6cm} i) \hspace{0.5cm} \bm{h} \gets \bm{h} - \mathbf{e}_{<k}
	&  (i) \hspace{0.31cm} \bm{h} \gets \bm{h} - \mathbf{e}_{k_<}
	\\
	%& \hspace{0.5cm} ii) \hspace{0.5cm} \bm{h} \gets \bm{h} + \mathbf{A} \bm{z} 
	& (ii) \hspace{0.2cm}  \bm{h} \gets \bm{h} + \mathbf{A} \bm{z} 
	\, .
	\end{split}
	\label{eq:iterate_hd_2}
	\end{align}
%In both steps the non-zero component in $\bm{z}$ must be drawn with respect to the excess degree distribution, i.e.\ $\bm{h}/h$.
where again the position $q$ where to place the one in the random variable $\bm{z}$ is chosen from the current excess degree distribution $h_q(n)  / \left| \bm{h} (n) \right|$ for $q= \, k_<, k,k+1, \dots, d$.	

The advantage of working with a representation in terms of half-edges is that we do not need to distinguish the different edges of color ``$k_<$''.
Further $\left| \bm{h}(n) \right|$ is deterministic because each column of (\ref{eq:relabeling_matrix}) sums to the same constant. During the removal step, eq.~(\ref{eq:iterate_hd_1}), we remove $2d$ half-edges and in one iteration of the trimming step, eq.~(\ref{eq:iterate_hd_2}), we remove $2$ half-edges (resp. $d$ edges and one edge). However, the running time of the second step has become a random variable. We have effectively traded the randomness in $\left| \bm{h}(n) \right|$ for randomness in the running time. For now we have simply shifted the problem into controlling the randomness in the running time. In section \ref{sec:operator_limit} it will be shown that transitioning to continuous time resolves this issue, after averaging, by determination of the running time $\delta t$ as a function of $\mathbb{E}\bm{h}(t)$.

This alternating process is also related to low-complexity algorithms for solving $K$-SAT problems~\cite{Achlioptas-stoc00,Achlioptas-tcs01,MonassonCocco01}.
These $K$-SAT solution methods alternate between guessing variables, which may create new unit clauses, and trimming unit clauses via unit clause propagation.
Due to this connection, the differential equation analyses for these two problems are somewhat similar.
 
\subsection{Taking the average over randomness}

As the equations stand in (\ref{eq:iterate_hd_1}) and (\ref{eq:iterate_hd_2}) they define a random process
that behaves just as Alg.~\ref{alg:hd} on a random graph $\mathcal{G}$, but with $\bm{z}$ and the stopping time of the trimming step implicitly containing all randomness. In terms of the urn with balls representing half-edges, the random variable $\bm{z}$ indicates the color of the second half-edge that is left behind after the first was removed.
We denote the average over $\bm{z}$ as
	\[
	\overline{\bullet} \coloneqq \mathbb{E}_{\bm{z}} \left[ \bullet \right]
	\, .
	\] 
Performing the average over the randomness per se only yields the average behaviour of the algorithm. In section \ref{subsec:rigorous} it is shown that the stochastic process concentrates around its average in the continuous limit.

Next the combination of steps $(i)$ \emph{and} $(ii)$ in eq.~(\ref{eq:iterate_hd_1}) for the ``removal'' and eq.~(\ref{eq:iterate_hd_2}) for the ``trimming'' is considered. 
In terms of half-edges we remove $2d$ half-edges in one iteration of (\ref{eq:iterate_hd_1}) and $2$ half-edges in one iteration of (\ref{eq:iterate_hd_2}). 
In order to write the average of the removal step, we recall that the probability that one half-edge is connected to a color $q\in \{ k_<,k,k+1,\dots,d\}$ is given by the excess degree distribution$h_q(n) / \left| \bm{h}(n) \right|$. 
In the large system limit the average drift of a full removal step can be written as
	\begin{align}
	\overline{\bm{h}}(n+1) 
	&= 
	\left( \bm{1} + \frac{1}{\left| \overline{\bm{h}}(n)\right|-d-(d-1)} \mathbf{A} \right) 
	\cdots 
	\left( \bm{1} + \frac{1}{\left| \overline{\bm{h}}(n)\right|-d} \mathbf{A} \right)
	\left( \overline{\bm{h}}(n) - d \, \mathbf{e}_d \right)
	\\
	&= 
	\left( \bm{1} + \sum_{j=0}^{d-1} \frac{1}{\left| \overline{\bm{h}}(n)\right|-d-j}  \mathbf{A} \right)
	\overline{\bm{h}}(n) 
	+ \left( \sum_{k=2}^{d} c_k \mathbf{A}^k  \right)	\overline{\bm{h}}(n) 
	-d\, \mathbf{e}_d - d \left( \sum_{k=1}^{d} \tilde{c}_k \mathbf{A}^k  \right) \mathbf{e}_d
	\\
	&= 
	\left( \bm{1} + \sum_{j=0}^{d-1} \frac{1}{\left| \overline{\bm{h}}(n)\right| - d -j }  \mathbf{A} \right)
	\overline{\bm{h}}(n) -d \, \mathbf{e}_d 
	+ O(\frac{1}{\left| \overline{\bm{h}}(n)\right|})
	\\
	&= 
	\left( \bm{1} + \frac{d}{\left| \overline{\bm{h}}(n)\right|}  \mathbf{A} \right)
	\overline{\bm{h}}(n) 
	- d \, \mathbf{e}_d
	+ O(\frac{1}{\left| \overline{\bm{h}}(n)\right|})
	\, ,
	\end{align}
where $\bm{1}$ represents the identity matrix.
In the above estimate, we use $d$ intermediate steps to transition from $n\to n+1$, that is the removal of a whole degree $d$ node.
We assume that $d$ is $O(1)$. It then follows that the coefficients $c_k$ and $\tilde{c}_k$ are $O(\left| \overline{\bm{h}}(n)\right|^{-k})$. The last line follows from a similar estimate for the leading term in the sum.

The average removal step can now be written as
	\begin{align}
	\overline{\bm{h}}(n+1) =  \overline{\bm{h}}(n) + \mathbf{A}_{d} \, \frac{\overline{\bm{h}}(n)}{{\left| \overline{\bm{h}}(n) \right|}}
	\, .
	\label{eq:averaged_hd_1}
	\end{align}
with the effective average drift matrix
	\begin{align}
	\mathbf{A}_{d} &\coloneqq d (\mathbf{A}+\mathbf{B}_d)
	\, ,
	\label{eq:drift_matrix_hd}
	\end{align}
where the matrix $\mathbf{B}_d$ has all entries in the last row equal to $-1$ and zeros everywhere else, such that $\mathbf{B}_d\bm{v}=-\mathbf{e}_d$ for a non-negative, normalized vector $\bm{v}$.
Similarly, taking the average in one trimming time step (\ref{eq:iterate_hd_2}) yields the following averaged version 
	\begin{align}
	\overline{\bm{h}}(n+1) = \overline{\bm{h}}(n) + \mathbf{A}_{k_<} \, \frac{\overline{\bm{h}}(n)}{{\left| \overline{\bm{h}}(n) \right|}}
	\, .
	\label{eq:averaged_hd_2}
	\end{align}
For the trimming step the effective drift is simply
	\begin{align}
	\mathbf{A}_{k_<} &\coloneqq  \mathbf{A}+\mathbf{B}_{k_<}
	\, .
	\label{eq:drift_matrix_trim}
	\end{align}
where now $\mathbf{B}_{k_<}$ has all its entries in the first row equal to $-1$ and zeros everywhere else.

We emphasize that the two processes (\ref{eq:averaged_hd_1}) and (\ref{eq:averaged_hd_2}), while acting on the same vector, are separate processes and the latter (\ref{eq:averaged_hd_2}) needs be repeated until the stopping condition $\overline{h}_{k_<}=0$ is hit. Note also, that in the trimming process, one iteration $n\to n+1$ indicates the deletion of a single edge, while it indicates the deletion of a whole node in the removal process.

\subsection{Operator and Continuous Limits \label{sec:operator_limit}}
\label{sec:oplimit}

As discussed at the end of Sec.~\ref{sec:reduction}, a key observation, by virtue of which we can proceed, is that $\left|{\bm{h}}(n)\right|$ is deterministic (and hence equal to its average) during both the removal and trimming steps. This is due to the structure of $\mathbf{A}_{k_<}$ and $\mathbf{A}_d$ that have columns sums independent of the row index:
\[
q_{k_<}\equiv \sum_{i} [\mathbf{A}_{k_<}]_{ij}= -2 
\hspace{0.5cm} \text{and} \hspace{0.5cm}
q_d \equiv \sum_{i} [\mathbf{A}_d]_{ij} = -2d
\]
The only randomness occurs in the stopping time of the trimming step.

In this section the transition to the continuous time-variable $t$ is performed. To that end we define the scaled process 
	\begin{equation}
	\boldsymbol{\eta}(t) = \lim_{N_0\to\infty}	\frac{1}{N_0} \, \overline{\bm{h}}(t N_0)
	\end{equation}
and presume that the derivative $\boldsymbol{\eta} ' (t)$ is equal to its expected change. Here $N_0$ stands for the initial number of vertices in the graph.

Before proceeding to the analysis of \coreHD, let us first describe the solution for the two processes (removal and trimming) as if they were running  separately. Let us indicate the removal process (\ref{eq:averaged_hd_1}) and trimming processes  (\ref{eq:averaged_hd_2}) with subscripts $\alpha = d$ and $\alpha = k_<$ respectively. It then follows from (\ref{eq:averaged_hd_1}) and (\ref{eq:averaged_hd_2}) that the expected change is equal to
\begin{equation}
	\boldsymbol{\eta}^{\boldsymbol{\prime}}_{\alpha}(t) = \mathbf{A}_{\alpha} \, \frac{\boldsymbol\eta_\alpha(t)}{\left| \boldsymbol\eta_\alpha (t) \right|}
	\label{eq:separate_ODE_hd}
	\, .
\end{equation}
Owing to the deterministic nature of the drift terms $A_\alpha$ we have
\begin{equation}
	{\left| \boldsymbol\eta_\alpha (t) \right|} = 1+q_\alpha t
\end{equation}  
and the above differential equation can be solved explicitly as 
\begin{equation}
	\boldsymbol\eta_\alpha(t) =  \exp\left[ \frac{\mathbf{A}_{\alpha}}{q_\alpha}  \ln \left( 1+q_\alpha t \right) \right]  \boldsymbol\eta_\alpha(0)
	\, .
	\label{eq:analytic_solution}
\end{equation}
We have thus obtained an analytic description of each of the two separate processes (\ref{eq:averaged_hd_1}) and (\ref{eq:averaged_hd_2}).

Note that this implies that we can analytically predict the expected value of the random process in which all nodes of degree $d$ are removed from a graph successively until none remains and then all nodes of degree smaller than $k$ are trimmed. This already provides improved upper bounds on the size of the minimal contagious sets, that we report in Table \ref{tbl:bounds} (cf.\ ``two stages''). This ``two stages'' upper bound has the advantage that no numerical solution of differential equations is required.
The goal, however, is to analyze the \coreHD\  procedure that merges the two processes into one, as this should further improve the bounds.

Crucially the running time of the trimming process depends on the final state of the removal process, i.e. the differential equations become nonlinear in $\boldsymbol\eta(t)$. As a consequence, they can no longer be brought into a simple analytically solvable form (at least as far as we were able to tell). 
To derive the differential equations that combine the removal and trimming processes and track \coreHD\ we will be working with the operators that are obtained from the ``iterations'', (\ref{eq:averaged_hd_1}) and (\ref{eq:averaged_hd_2}), in the continuous limit (\ref{eq:separate_ODE_hd}). 
The evolution within an infinitesimally small removal step ($\alpha=d$), respectively trimming step ($\alpha=k_<$), follows from (\ref{eq:separate_ODE_hd}) to 
	\begin{equation}
	\boldsymbol\eta_{\alpha}(t+\delta t) = 
	\hat{\mathbf{T}}_\alpha (\delta t,t) 
	%\left( \bm{1}+ \frac{ \mathbf{A}_\alpha }{1-q_\alpha t} \, \delta t \right)
	\boldsymbol\eta_{\alpha}(t)
	\, ,
	\label{eq:dt_step}
	\end{equation}
where we defined the propagator 
	\[
	\hat{\mathbf{T}}_\alpha (\delta t,t) = \left( \bm{1}+ \frac{ \mathbf{A}_\alpha }{\left| \boldsymbol\eta_\alpha (t) \right|} \, \delta t \right)
	\, .
	\]

In what follows we will be considering the removal and trimming processes to belong to one and the same process and therefore $\boldsymbol\eta(t)$ will no longer be carrying subscripts.  Upon combination a full step in the combined process in terms of the operators then reads 
	\begin{equation}
	\boldsymbol\eta\left( t+\delta t \right) = \hat{\mathbf{T}}_{k_<} (\hat{\delta t},t+\delta t) \, \hat{\mathbf{T}}_d (\delta t,t) \, \boldsymbol\eta \left( t \right)
	\, .
	\label{eq:propagator}
	\end{equation}
Note that one infinitesimal time step is the continuous equivalent of the removal of one degree $d$ node, together with the resulting cascade of degree $<k$ nodes. It is for that reason that the final continuous time after which the $k$-core vanishes will be directly related to the size of the set $\mathcal{D}$ in Algorithm~\ref{alg:hd}. 
Note also that in $\hat{\mathbf{T}}_{k_<}(\hat{\delta t},t+\delta t)$ we replaced the running time with the operator $\hat{\delta t}$. It acts on a state to its right and can be computed from the condition that \emph{all} the nodes of degree smaller $k$ must be trimmed after completing a full infinitesimal step of the combined process, so that
	\begin{equation}
	\eta_{k_<}\left( t+\delta t \right) \overset{!}{=} 0
	\, .
	\label{eq:nonlinear_condition}
	\end{equation}
Requiring this condition in eq.~(\ref{eq:propagator}) we get from an expansion to linear order in $\delta t$ that 
	\begin{equation}
	\hat{\delta t}  = \delta t \cdot \left(  - \frac{\left[ \mathbf{A}_d \boldsymbol\eta(t) \right]_{k_<}}{\left[ \mathbf{A}_{k_<} \boldsymbol\eta(t) \right]_{k_<}} \right)
	\, .
	\label{eq:running_time}
	\end{equation}

Once again, we recall that $[\bm{v}]_{k_<}$ denotes the first component of the vector $\bm{v}$.
We can now use this equation to eliminate the dependence on $\hat{\delta t}$ in the combined operator $\hat{\mathbf{T}}_{k_<} (\hat{\delta t},t) \, \hat{\mathbf{T}}_d (\delta t,t)$. 
Using (\ref{eq:dt_step}) and keeping only first order terms in $\delta t$ in (\ref{eq:propagator}) yields
	\begin{equation}
	\boldsymbol\eta(t+\delta t) = \boldsymbol\eta(t) + \left[   - \frac{\left[ \mathbf{A}_d \boldsymbol\eta(t) \right]_{k_<}}{\left[ \mathbf{A}_{k_<} \boldsymbol\eta(t) \right]_{k_<}} \, \mathbf{A}_{k_<} + \mathbf{A}_d \right] \, \frac{\boldsymbol\eta(t)}{\left| \boldsymbol\eta (t) \right|} \, \delta t,
	\label{eq:propagator_contiuous}
	\end{equation}
which leads us to the following differential equation
	\begin{equation}
	\boldsymbol{\eta'}(t) = \left[ \varphi\left( \boldsymbol\eta(t) \right) \mathbf{A}_{k_<} + \mathbf{A}_d \right] \frac{\boldsymbol\eta(t)}{\left| \boldsymbol\eta (t) \right|}
	\, .
	\label{eq:non_linear_dgl}
	\end{equation}
The nonlinearity $\varphi(\cdot)$ is directly linked to the trimming time and defined as 	
	\begin{equation}
	\varphi\left( \boldsymbol\eta(t) \right) \equiv - \frac{\left[ \mathbf{A}_d \boldsymbol\eta(t) \right]_{k_<}}{\left[ \mathbf{A}_{k_<} \boldsymbol\eta(t) \right]_{k_<}}
	=
	-
	\frac{d\left( -\eta_{k_<}(t) +(k-1)\eta_k(t) \right)}{-\left| \boldsymbol\eta (t) \right|-\eta_{k_<}(t)+(k-1)\eta_k(t)}
	=
	\frac{d\, (k-1) \eta_k(t)}
	{\left| \boldsymbol\eta (t) \right|-(k-1)\eta_k(t)}
	\, .
	\label{eq:nonlinear_term}
	\end{equation}
To obtain the last equality in (\ref{eq:nonlinear_term}) we used the trimming condition, i.e. set $\eta_{k_<}(t)=0$. 
The initial conditions are such that the process starts from the $k$-core of the original graph. This is achieved by solving (\ref{eq:analytic_solution}), with $\alpha=k_<$, for arbitrary initial degree distribution $\boldsymbol{\eta}(0)$ (without bounded maximum degree) until $\eta_{k_<}(t)=0$. 
Hence, the set of differential equations defined by (\ref{eq:non_linear_dgl}) can be written explicitly as
	\begin{eqnarray}
	\eta_{k_<}'(t) &=& 0 \label{eq:dgl_final1}
	\\
	\eta_{i}'(t) &=& d\cdot i \cdot \frac{- \eta_i (t) + \eta_{i+1} (t)}{\left| \boldsymbol{\eta}(t) \right| -(k-1)\eta_k (t)}
	\hspace{0.8cm} \textrm{for } k \leq i<d 
	\\
	\eta_d'(t) &=&-d  + d^2 \cdot \frac{- \eta_d (t) }{\left| \boldsymbol{\eta}(t) \right| -(k-1)\eta_k (t)}
	\, .
	\label{eq:dgl_final3}
	\end{eqnarray}

\subsection{Rigorous Analysis }
\label{subsec:rigorous}

A rigorous analysis of the $k$-core peeling process for Erd\H{o}s-R\'{e}nyi graphs is presented in~\cite{Pittel-jctb96}.
This analysis is based on the Wormald approach~\cite{Wormald1995} but the presentation in~\cite{Pittel-jctb96} is more complicated because it derives an exact formula for the threshold and there are technical challenges as the process terminates.
For random graphs drawn from the configuration model, however, the standard Wormald approach~\cite{Wormald1995} provides a simple and rigorous numerical method for tracking the macroscopic dynamics of the peeling algorithm when the maximum degree is bounded and the degree distribution remains positive.
The primary difficulty occurs near termination when the fraction of degree $<\!k$ edges becomes very small. % but, as we will see, this can be handled separately

The peeling process in \coreHD\ alternates between deleting maximum-degree nodes and degree $<\!k$ edges and this introduces a similar problem for the Wormald method.
In particular, the \coreHD\  peeling schedule typically reduces the fraction of maximum-degree nodes to zero at some point and then the maximum degree jumps downward.
At this jump, the drift equation is not Lipschitz continuous and does not satisfy the necessary conditions in~\cite{Wormald1995}.
More generally, whenever there are hard preferences between node/edge removal options (i.e., first delete largest degree, then 2nd largest degree, etc.), the same problem can occur.

For \coreHD, one solution is to use weighted preferences where the removal of degree $<\!k$ edges is most preferred, then removal of degree-$d$ nodes, then degree $d-1$ nodes, and so on.
In this case, the drift equation remains Lipschitz continuous if the weights are finite but the model dynamics only approximate the \coreHD\  algorithm dynamics.
In theory, one can increase the weights to approximate hard preferences but, in practice, the resulting differential equations become too numerically unstable to solve efficiently.
A better approach is to use the operator limit described in Section~\ref{sec:oplimit}.
Making this rigorous, however, requires a slightly more complicated argument.

The key argument is that the $k$-core peeling step (after each maximum-degree node removal) does not last too long or affect too many edges in the graph.
A very similar argument (dubbed the Lazy-Server Lemma) is used in the analysis of low-complexity algorithms for solving $K$-SAT problems~\cite{Achlioptas-stoc00,Achlioptas-tcs01}.
In both cases, a suitable stability (or drift) condition is required.
In this work, we use the following lemma.
%, which is equivalent to requiring that~\eqref{eq:dgl_finala} and~\eqref{eq:dgl_final} are finite.
\begin{lem}\label{lem:stoptime}
For some $\delta>0$, suppose $\bm{h}(n)$ satisfies $h_k (n) \leq \frac{1-3\delta}{k-1} |\bm{h}(n)|$ and $M \triangleq |\bm{h}(n)|\geq M_0(\delta)$.
Consider the \coreHD\  process where a maximum-degree node is removed and then the trimming operation continues until there are no edges with degree less than $k$ (see~\eqref{eq:iterate_hd_2}).
Let the random variable $T$ denote the total number of trimming steps, which also equals the total number edges removed by the trimming operation.
Then, we have
\[\Pr \big(T > (k-1)^2 \delta^{-2} \ln M \big) \leq M^{-2}.\]

\end{lem}

\begin{proof}
See Appendix~\ref{appendix_lem_stoptime_proof}.
\end{proof}

\begin{lem}
\label{lem:diffeq}
Let $\boldsymbol{\eta}(t)$ be the solution to the operator-limit differential equation~\eqref{eq:non_linear_dgl} at time $t$ starting from $\boldsymbol{\eta}(0) = \frac{1}{N} \bm{h}(0)$.
Assume, for some $\delta>0$, there is a $t_0 > 0$ such that\footnote{The first condition implies that the denominator of~\eqref{eq:nonlinear_term} is positive.} $\eta_k (t) \leq \frac{1-4\delta}{k-1} \left| \boldsymbol\eta (t) \right|$ and $\left| \boldsymbol\eta (t) \right|\geq \delta$ for all $t\in[0,t_0]$.
Then, there is $C>0$ such that
\[\Pr \left( \sup_{t\in[0,t_0]} \left\| \boldsymbol{\eta}(t) - \frac{1}{N} \bm{h} (\lfloor Nt \rfloor) \right\| > \frac{C}{N^{1/4}} \right) = O \left( N^{-1}\right).\]
\end{lem}

\begin{proof}
See Appendix~\ref{appendix_lem_diffeq_proof}.
\end{proof}

\begin{thm}
The multistage \coreHD\  process converges, with high probability as $N\to \infty$, to the piecewise solution of the operator-limit differential equation.
\end{thm}

\begin{proof}[Sketch of Proof]
The first step is recalling that the standard $k$-core peeling algorithm results in graph distributed according to the configuration model with a degree distribution that, with high probability as $N\to\infty$, converges to the solution of the standard $k$-core differential equation~\cite{Pittel-jctb96}.
If $k$-core is not empty, then the \coreHD\  process is started.
To satisfy the conditions of~\cite[Theorem~5.1]{Wormald-lara99}, the process is stopped and restarted each time the supply of maximum-degree nodes is exhausted.
Since the maximum degree is finite, this process can be repeated to piece together the overall solution.
Using Lemma~\ref{lem:diffeq}, we can apply~\cite[Theorem~5.1]{Wormald-lara99} at each stage to show the \coreHD\  process follows the differential equation~\eqref{eq:non_linear_dgl}.
It is important to note that the cited theorem is more general than the typical fluid-limit approach and allows for unbounded jumps in the process as long as they occur with low enough probability.  
\end{proof}

\subsection{Evaluating the results}

Here we clarify how the upper bound is extracted from the equations previously derived. 
Note that the nonlinearity (\ref{eq:nonlinear_term}) exhibits a singularity when 
	\begin{equation}
	\left| \boldsymbol\eta (t) \right|=(k-1)\eta_k(t)
	\, ,
	\label{eq:stopping_condition}
	\end{equation}
that is, when the gain (r.h.s.) and loss (l.h.s.) terms in the trimming process are equal. This can be either trivially true when no more nodes are left, $\left| \boldsymbol\eta (t) \right|=0$, or it corresponds to an infinite trimming time.
The latter is precisely the point where the size of the $k$-core jumps downward discontinuously, whereas the first case is linked to a continuous disappearance of the $k$-core. Either of these two cases define the \emph{stopping time} $t_{\textrm{s}}$ of the differential process (\ref{eq:non_linear_dgl}). By construction the stopping time $t_{\textrm{s}}$ provides the size of the set $\mathcal{D}$ that contains all the nodes the \coreHD\ algorithm removed to break up the $k$-core. It hence also \emph{provides an upper bound} on the size of the minimal contagious set, i.e.\ the smallest such set that removes the $k$-core. 

Note that $\boldsymbol\eta(t_{\textrm{s}}-\epsilon)$ (for an infinitesimally small $\epsilon$) gives the size of the $k$-core, right before it disappears. For all the cases investigated in this paper we found that solving eqs.~(\ref{eq:dgl_final1}--\ref{eq:dgl_final3}) for $k=2$ yields a continuous disappearance of the $2$-core, and for $k\ge 3$ the stopping criteria yield discontinuous disappearance of the $k$-core.

\begin{algorithm}[H] 
	\SetAlgoNoLine
	\LinesNumbered
		\KwData{Initial degree distribution $Q(q)$; $k$}
		\KwResult{The relative size of set of removed nodes, $t_{\text{s}}$.}
		\textbf{Initialize half-edges}: 
		$\boldsymbol{\eta}(0) = 
		\left(\sum_{i=1}^{k-1} i\cdot Q(i),k\cdot Q(k), (k+1)\cdot Q(k+1),\hdots,d\cdot Q(d) \right)^T$
		\;
		\textbf{Compute distribution of half-edges in the $k$-core}: 
		$\boldsymbol\eta_0 = \exp\left[ -\frac{\mathbf{A}_{k_<}}{2} \cdot \ln \left( 1-2 t_0 \right) \right] \boldsymbol\eta(0)$  with $t_0$ such that $\eta_{k_<}(t_0) = 0$ and $\mathbf{A}_{k_<}$ defined in (\ref{eq:drift_matrix_trim})\;
		\textbf{Set}: $d \gets$ the degree associated to the last non-zero component of $\boldsymbol{\eta}_0$; $t_{\text{s}} \gets 0$  \;
		\While{$\textrm{stop} \neq \textrm{true}$}{
			\textbf{Solve}: 
			\[
				\boldsymbol{\eta'}(t) = \left[ \frac{d(k-1)\eta_k(t)}
				{\left| \boldsymbol\eta (t) \right|-(k-1)\eta_k(t)} \mathbf{A}_{k_<} + \mathbf{A}_d \right] \frac{\boldsymbol\eta(t)}{\left| \boldsymbol\eta (t) \right|}
			\] 
			with initial condition $\boldsymbol{\eta}_0$. Until either $\eta_d(t^{*})=0$ or $\left| \boldsymbol\eta (t^{*})\right|=(k-1)\eta_k(t^{*})$
			\tcp*{$\mathbf{A}_{k_<}$ and $\mathbf{A}_d$ defined by (\ref{eq:drift_matrix_trim}) and (\ref{eq:drift_matrix_hd}) with $d$ set to the current largest degree}
			\textbf{Update}: $t_{\text{s}}\gets t_{\text{s}}+t^{*}$ \;
				\eIf{$\left| \boldsymbol\eta (t^{*}) \right|=0$ or $\left| \boldsymbol\eta (t^{*}) \right|=(k-1)\eta_k(t^{*})$}{
				stop $\gets$ true\;
				}{
				$\boldsymbol\eta_0 \gets \eta_{k_{<}:d-1}(t^{*})$ \;
				$d \gets d-1$ \;
				}
			}
		\caption{Analysis of \coreHD. Recall that indices are $k_<$ are referring to the first component of a vector, $k$ to the second and so forth until the last component $d$. }
		\label{alg:procedure}
	\end{algorithm}

In order to solve the above set of ODEs numerically, we first use equation (\ref{eq:analytic_solution}) to trim away nodes of color $k_<$, i.e. reduce the graph to it's $k$-core. Then we use equation (\ref{eq:propagator_contiuous}) recursively, until the last component $\eta_d$ is zero. Subsequently we reduce $\boldsymbol \eta$ by removing its last component, send $d\to d-1$, adapt the drift term (\ref{eq:drift_matrix_hd}) and repeat with the reduced $\boldsymbol \eta$ and initial condition given by the result of the previous step. All this is performed until the stopping condition (\ref{eq:stopping_condition}) is reached. We summarize the procedure in a pseudo-code in Algorithm~\ref{alg:procedure}, for our code that solves the differential equations see open depository \cite{demo_git_Christian}.

\paragraph*{Example: two-cores on three regular random graphs.}  

For a simple example of how to extract the upper bound consider the following case.
We have $k=2$ and $d=3$ and we set $k_<=1$, then the differential equation in (\ref{eq:non_linear_dgl}) becomes
	\begin{eqnarray}
	\eta^{\prime}_1(t) &=& 0 \\
	\eta^{\prime}_2(t) &=& 6 \cdot \left( 1 - \frac{\eta_2(t)}{\eta_3(t)} \right) \\
	\eta^{\prime}_3(t) &=& -12
	\, ,
	\label{eqn:example}
	\end{eqnarray}
with initial condition $\boldsymbol{\eta}(0)=d\cdot \mathbf{e}_d$ because there are $d N$ half-edges, all connected to nodes of degree $d$ initially. The equations are readily solved 
	\begin{eqnarray}
	\eta_1(t) &=& 0 \\
	\eta_2(t) &=& 3 \cdot \sqrt{1-4t} \cdot \left( 1- \sqrt{1-4t} \right) \\
	\eta_3(t) &=& 3\cdot (1 - 4t)
	\, ;	
	\end{eqnarray}
According to (\ref{eq:stopping_condition}) the stopping time is $t_s = 1/4$, i.e. $\left| \mathcal{D} \right|=N/4$, which suggests that the decycling number ($2$-core) for cubic random graphs is bounded by $N/4 + o(N)$. In accordance with Theorem 1.1 in~\cite{BauWormaldZhou02} this bound coincides with the actual decycling number. For $d>3$ the lower and upper bounds do not coincide an the stopping time resulting from our approach only provides an upper bound. 
\\

Finding a closed form expression for the generic case is more involved and we did not manage to do it. However, very reliable numerical resolution is possible. The simplest approach to the differential equations is to work directly with (\ref{eq:propagator_contiuous}) as indicated in Algorithm~\ref{alg:procedure}.

\subsection{\coreHD\  analyzed and compared with existing results}

In this section we evaluate the upper bound on minimal contagious set obtained by our analysis of \coreHD.  In Figure \ref{fig:compare_dyn_fig} we compare the fraction of nodes of a given degree that are in the graph during the \coreHD procedure obtained from solving the differential equations and overlay them with averaged timelines obtained from direct simulations of the algorithm.  

\begin{figure}
	\begin{center}
	\includegraphics[scale=1]{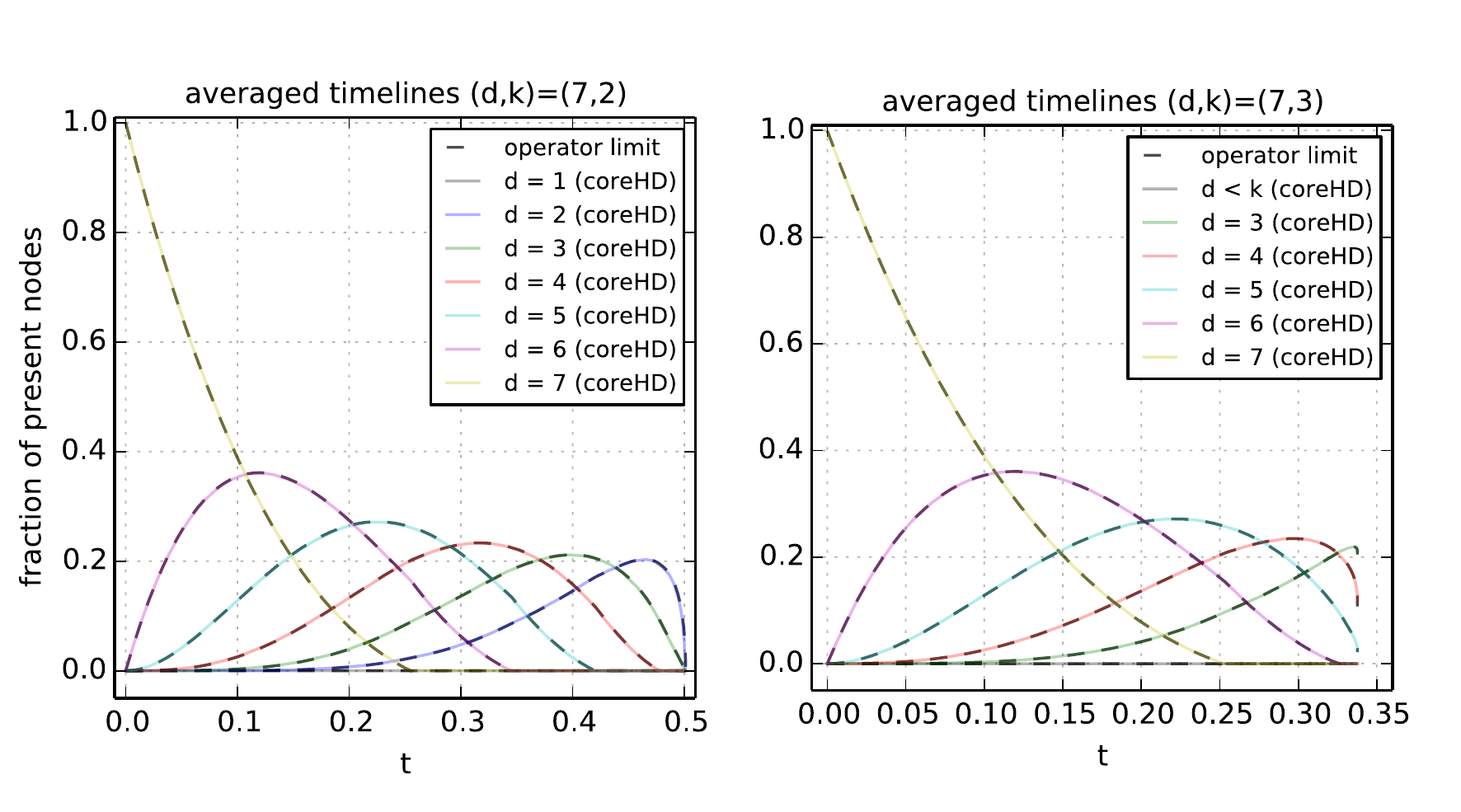}
	\end{center}
	\caption{Comparison of the solution of the differential equation Algorithm~\ref{alg:procedure} with a direct simulations of the \coreHD\  algorithm on several graphs of $2^{14}$ nodes, here with $d=7$ and $k=2$ (left hand side) and $k=3$ (right hand side). 
	The time $t$ is directly related to the fraction of removed nodes of degree $\geq k$. For $k=2$ the $k$-core disappears continuously at $t_s=0.5006$, for $k=3$ discontinuously at $t_s=0.3376$.
	}
	\label{fig:compare_dyn_fig}
	\end{figure}

Table \ref{tbl:compare_dgl_sim} then compares direct numerical simulations of the \coreHD\  algorithm with the prediction that is obtained from the differential equations. The two are in excellent agreement, for the analysis higher precision can be obtained without much effort. When analyzing Erd\H{o}s-R\'{e}nyi graphs it is necessary to restrict the largest degree to avoid an infinite set of differential equations. The error resulting from that is exponentially small in the maximum degree, and hence in practice insignificant.

	\begin{table}
	\centering
	\begin{tabular}{|c|c||c|c|c||c|c|c|c|c|}
            \hline 
            $d$ & $k$ & RRG \coreHD & RRG theory & final core & ERG \coreHD & ERG theory & final core \tabularnewline
            \hline 
            \hline 
          
            $3$ & $2$ & $.2500(0)$ & $.25000$ & $0$ & $.1480(7)$ & $.14809$ & $0$ \tabularnewline
          
            \hline 
          
            $4$ & $2$ & $.3462(3)$ & $.34624$ & $0$ & $.2263(0)$ & $.22634$ & $0$ \tabularnewline
            \hline 
             & $3$ & $.0962(5)$ & $.09623$ & $.67349$ & $.0388(5)$ & $.03887$ & $.32718$ \tabularnewline
          
            \hline 
          
            $5$ & $2$ & $.4110(7)$ & $.41105$ & $0$ & $.2924(8)$ & $.29240$ & $0$ \tabularnewline
            \hline 
             & $3$ & $.2083(6)$ & $.20832$ & $.48428$ & $.1068(2)$ & $.10679$ & $.33105$ \tabularnewline
            \hline 
             & $4$ & $.0476(9)$ & $.04764$ & $.85960$& $-$ & $-$ & $-$ \tabularnewline
            
            \hline 
            
            $6$ & $2$ & $.4606(2)$ & $.46063$ & $0$ & $.3480(8)$ & $.34816$ & $0$ \tabularnewline
            \hline 
             & $3$ & $.2811(2)$ & $.28107$ & $.40916$ & $.1700(1)$ & $.16994$ & $.34032$ \tabularnewline
            \hline 
             & $4$ & $.1401(1)$ & $.14007$ & $.68672$ & $.0466(0)$ & $.04662$ & $.46582$ \tabularnewline
            \hline 
             & $5$ & $.0280(7)$ & $.02809$ & $.92410$ & $-$ & $-$ & $-$ \tabularnewline
             
            \hline 
            
            $7$ & $2$ & $.5006(0)$ & $.50060$ & $0$ & $.3954(7)$ & $.39554$ & $0$ \tabularnewline
            \hline 
             & $3$ & $.3376(1)$ & $.33757$ & $.36179$ & $.2260(0)$ & $.22597$ & $.32878$ \tabularnewline
            \hline 
             & $4$ & $.2115(0)$ & $.21150$ & $.58779$ & $.1043(2)$ & $.10429$ & $.46788$ \tabularnewline
            \hline 
             & $5$ & $.1010(2)$ & $.10100$ & $.78903$ & $.0088(5)$ & $.00882$ & $.54284$ \tabularnewline
            \hline 
             & $6$ & $.0184(6)$ & $.01842$ & $.95038$ & $-$ & $-$ & $-$ \tabularnewline
             
            \hline 
	\end{tabular}
	\caption{Comparison between direct simulations of the \coreHD\  Algorithm \ref{alg:hd} and the theoretical results for random graphs of different degree $d$ and core-index $k$. The left hand side of the table reports results for random regular graphs, the right hand side for Erd\H{o}s-R\'{e}nyi random graphs of the corresponding average degree.  The simulation results are averaged over $50$ runs on graphs of size $N=2^{19}$. The fractions of removed nodes agree up to the last stable digit (digits in brackets are fluctuating due to finite size effects). We also list the size of the $k$-core just before it disappears (in the column ``final core'').}
	\label{tbl:compare_dgl_sim}
	\end{table}

Confident with this cross-check of our theory, we proceed to compare with other theoretical results. As stated in the introduction, Guggiola and Semerjian~\cite{GuggiolaSemerjian2015} have derived the size of the minimal contagious sets for random regular graphs using the cavity method. At the same time, several rigorous upper bounds on size of the minimal contagious set exist~\cite{Ackerman2010,Reichman2012,Coja-Oghlan14}.
In particular the authors of~\cite{BauWormaldZhou02} provide upper bounds for the decycling number ($2$-core) that are based on an analysis similar to ours, but of a different algorithm.\footnote{The numerical values, provided for the bounds in~\cite{BauWormaldZhou02} are actually not correct, as the authors realized and corrected in a later paper~\cite{HoppenWormald08}. We acknowledge the help of Guilhem Semerjian who pointed this out.} In table \ref{tbl:bounds} we compare the results from~\cite{BauWormaldZhou02} with the ones obtained from our analysis and the presumably exact results from~\cite{GuggiolaSemerjian2015}. We clearly see that while \coreHD\  is not quite reaching the optimal performance, yet the improvement over the existing upped bound is considerable. 
    \begin{table}
    \centering
        \begin{tabular}{|c||c|c|c|c|}
            \hline 
            $d$ (degree) & Bau, Wormald, Zhou & two stages & \coreHD & cavity method \tabularnewline
            \hline 
            \hline 
            $3$ & $.25$ & $.3750$ & $.25$ & $.25$ \tabularnewline
            \hline 
            $4$ & $.3955$ & $.3913$ & $.3462$ & $.3333$ \tabularnewline
            \hline 
            $5$ & $.4731$ & $.4465$ & $.4110$ & $.3785$ \tabularnewline
            \hline 
            $6$ & $.5289$ & $.4911$ & $.4606$ & $.4227$ \tabularnewline
            \hline 
            $7$ & $.5717$ & $.5278$ & $.5006$ & $.4602$ \tabularnewline
            \hline 
        \end{tabular}
    \caption{Best known upper bounds on the minimal decycling sets by~\cite{BauWormaldZhou02}, compared to the upper bounds obtained from our analysis when all nodes of maximum degree are removed before the graph is trimmed back to its $2$-core (two stages), and to our analysis of \coreHD\ . The last column gives  the non-algorithmic cavity method results of Guggiola and Semerjian~\cite{GuggiolaSemerjian2015} that provide (non-rigorously) the actual optimal sizes. }
    \label{tbl:bounds}
    \end{table}

In table \ref{tbl:compare_algorithms} we quantify the gap between our upper bound and the results of~\cite{GuggiolaSemerjian2015} for larger values of $k$. Besides its simplicity, the \coreHD\  algorithm provides significantly better upper bounds than those known before. Clearly, we only consider a limited class of random graphs here and the bounds remain away from the conjectured optimum. However, it is worth emphasizing that previous analyses were often based on much more involved algorithms. The analysis in~\cite{Coja-Oghlan14} or the procedure in~\cite{BauWormaldZhou02} are both based on algorithms that are more difficult to analyze.

\section{Improving \coreHD}\label{sec:improving}

\begin{figure}
	\begin{center}
		\includegraphics[scale=0.9]{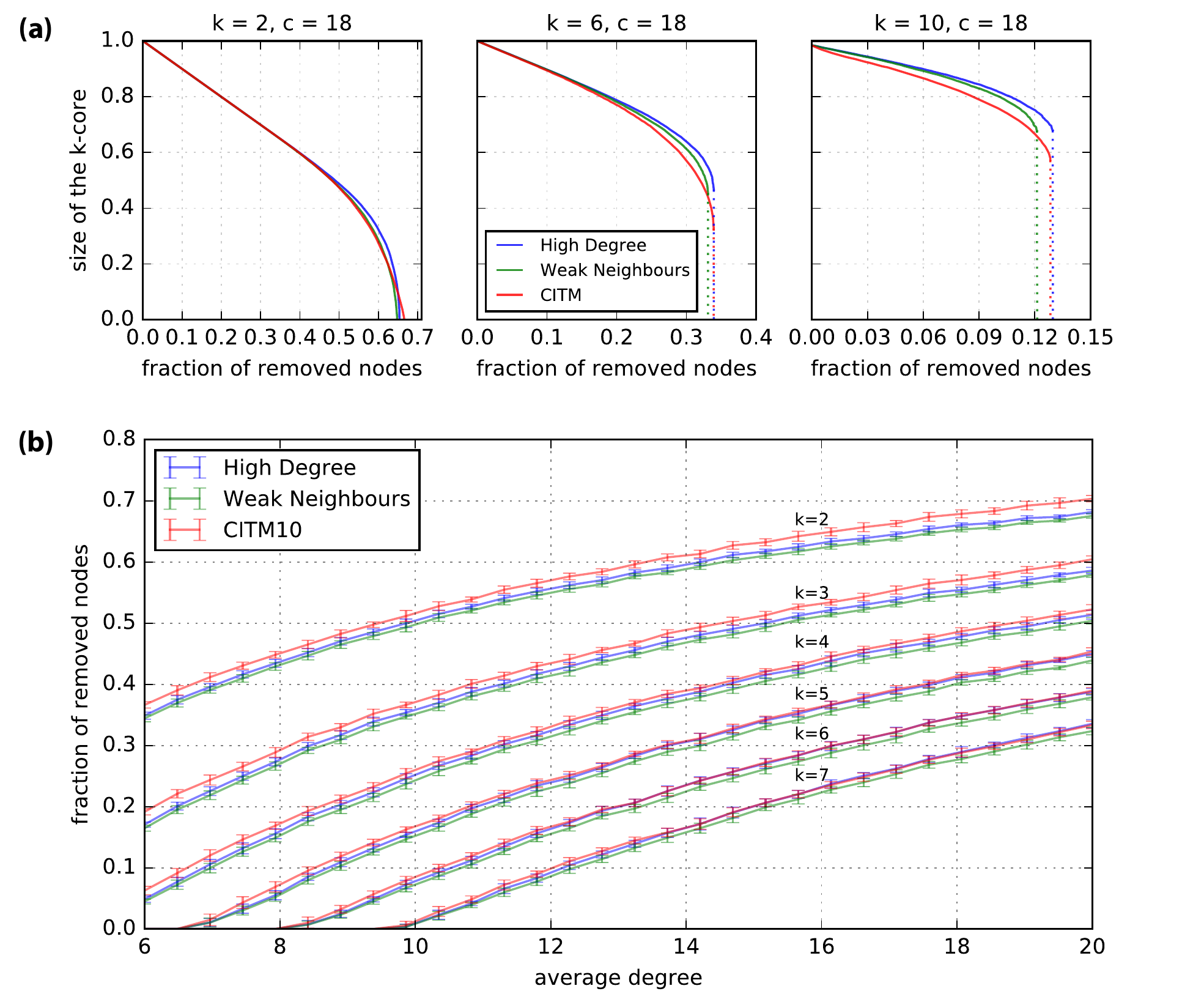}
	\end{center}
	\caption{(a) A typical transition of three algorithms, with $L$ for the \CITM\ set to the size of the graph. When $k$ is increased the gap between the \CITM\ algorithm and the \coreHD\  algorithm closes and they become close in terms of performance. However, there remains a gap to the \WN\   algorithm. 
		(b) Comparison of the three algorithms (here \CITM\  with $L=10$) on ER graphs for different $k$ and vary the average degree. Compare also table \ref{tbl:compare_algorithms} for results on regular random graphs. 
	}
	\label{fig:compare_er}
\end{figure}

The main focus of this paper has been, up until now, the analysis of \coreHD\ on random networks. 
Rather naturally the question of how to improve over it raises. In this section we evaluate the possibility of a simple local strategy that systematically improves over the \coreHD\  performance. We show that introducing additional local information about the direct neighbors of a node into the decision process can significantly improve the performance, while essentially conserving the time complexity. 

The \coreHD\  algorithm (Alg.~\ref{alg:hd}) does not take into account any information of the neighborhood of a node. The theoretical analysis in the previous section owes its simplicity to this fact. However, the idea behind \coreHD\  can be extended to the next order by considering the structure of the nearest neighbors of a node. Once we include the nearest neighbors the number of options is large. Our aim is not to make an extensive study of all possible strategies, but rather to point out some heuristic arguments that lead to improvement.

According to the previous section, selecting high degree nodes is a good strategy. Another natural approach is a greedy strategy that, in each step, selects a node such that the caused cascade of nodes dropping out of the core is maximized~\cite{GuggiolaSemerjian2015,Coja-Oghlan14,MakseMorone2017}.
In the following we list some strategies that aim to somehow combine these two on the level of the direct neighborhood of a node. For graphs with $O(1)$ maximum degree they can all be implemented in $O(N)$ running time.
\\

\paragraph*{\textsc{corehdld}:}
This approach selects high degree nodes, but then discriminates those that have many neighbors of high degree. The idea is that nodes that have neighbors of large degree might get removed in the trimming procedure preferentially and hence the degree of the node in question will likely decrease. More specifically we implemented the following: 
\begin{itemize}
\item First select the set $\mathcal{M} \gets \{i \in \mathcal{V}_{\mathcal{C}} \mid i = \arg\max_i \left[ d_i \right]\}$ and then update it according to $\mathcal{M} \gets \{i \in \mathcal{M} \mid i = \arg\max_i \left[ \sum_{j\in \partial i}\mathbb{I}\left[d_j<d_i\right] \right] \}$.
\end{itemize}

\paragraph*{\WN:}
The \WN\  strategy aims to remove first those nodes that have high degree \emph{and} low average degree of the neighbors, thus causing a larger trimming step on average. There are different ways to achieve this. We tried two strategies that both yield very similar results on all the cases considered. These two strategies are
\begin{itemize}
\item The order in which the nodes are removed is according to $\arg\max_i d_i - s_i$ with $d_i$ being the degree of node $i$ and $s_i$ the average degree of the neighbors of $i$. 
\item We separate the two steps. First select the set $\mathcal{M} \gets \{i \in \mathcal{V}_{\mathcal{C}} \mid i = \arg\max_i \left[ d_i \right]\}$ and then update it according to $\mathcal{M} \gets \{i \in \mathcal{M} \mid i = \arg\min_i \left[ \sum_{j\in \partial i}d_j \right] \}$.
\end{itemize}
Our implementation of the \WN\ algorithm is available in the open depository \cite{demo_git_Christian}.

\paragraph*{\textsc{corehd-critical}:}
The \textsc{coreHD-critical} combines the \coreHD\  with the vanilla-greedy algorithm on the direct neighbors. Nodes are first selected according to their degree and subsequently among them we remove nodes first that have the largest number of \emph{direct neighbors} that will drop out in the trimming process. 
\begin{itemize}
\item First select the set $\mathcal{M} \gets \{i \in \mathcal{V}_{\mathcal{C}} \mid i = \arg\max_i \left[ d_i \right]\}$ and then update it according to $\mathcal{M} \gets \{i \in \mathcal{M} \mid i = \arg\max_i \left[ \sum_{j\in \partial i}\mathbb{I}\left[d_j\leq k\right] \right] \}$.
\end{itemize}
Finally it is interesting to contrast the stated algorithms with a version in which the high degree selection step is left out, i.e. select $\mathcal{M} \gets \{i \in \mathcal{V}_{\mathcal{C}} \mid i = \arg\max_i \left[ \sum_{j\in \partial i}\mathbb{I}\left[d_j\leq k\right] \right]\}$ and then remove at random from this set.

\begin{table}[H]
	\centering
	\begin{tabular}{|c|c||c|c|c|c|c||c}
		\hline 
		$d$ & $k$ & \CITM-10 & \coreHD & \WN & cavity method \tabularnewline
		\hline 
		\hline 
		
		$3$ & $2$ & $.254$ & $.2500$ & $.2500$ & $.25000$\tabularnewline
		
		\hline 
		
		$4$ & $2$ & $.365$ & $.3462$ & $.3376$ &$.33333$\tabularnewline
		\hline 
		& $3$ & $.077$ & $.0963$ &$.0744$ & $.04633$\tabularnewline
		
		\hline 
		
		$5$ & $2$ & $.437$ & $.4111$ &$.3965$ & $.37847$\tabularnewline
		\hline 
		& $3$ & $.205$ & $.2084$ &$.1876$ & $.16667$\tabularnewline
		\hline 
		& $4$ & $.032$ & $.0477$ & $.0277$ & $.01326$\tabularnewline
		
		\hline
		
		$6$ & $2$ & $.493$ & $.4606$ &$.4438$ & $.42270$\tabularnewline
		\hline 
		& $3$ & $.296$ & $.2811$ & $.2644$ &$.25000$\tabularnewline
		\hline 
		& $4$ & $.121$ & $.1401$ & $.1081$ &$.07625$\tabularnewline
		\hline 
		& $5$ & $.019$ & $.0281$ &$.0134$ & $.00582$\tabularnewline
		
		\hline 
		
		$7$ & $2$ & $.540$ & $.5006$ & $.4831$ & $.46023$\tabularnewline
		\hline 
		& $3$ & $.362$ & $.3376$ &$.3206$ & $.30009$\tabularnewline
		\hline 
		& $4$ & $.207$ & $.2115$ &$.1813$ & $.15006$\tabularnewline
		\hline 
		& $5$ & $.081$ & $.1010$ & $.0686$ & $.04289$\tabularnewline 
		\hline 
		& $6$ & $.013$ & $.0185$ & $.0077$ & $.00317$\tabularnewline
		
		\hline 
	\end{tabular}
	\caption{Performance results of three algorithms: \CITM\  (with $L=10$), \coreHD, \WN\ (with $\arg\max_i d_i-s_i$). As well as the conjectured optimal results obtained non-constructively via the cavity method in~\cite{GuggiolaSemerjian2015}.}
	\label{tbl:compare_algorithms}
\end{table}

Let us summarize the results. First, we find that all above strategies improve over the \coreHD\  algorithm (at least in some regimes). Second, we find that among the different strategies the \WN\ algorithm performs best.

While we have systematic numerical evidence that the \WN\  strategy performs best, it is not clear which are the effects responsible. What we can say for sure is that the generic locally greedy procedure of trying to reduce the size of the $k$-core at every step is not optimal.

The most commonly considered greedy procedures do consider information not only from the direct neighborhood of a node. The vanilla-greedy approach removes nodes according to the size of cascade that is caused by their removal.\footnote{Note that such greedy algorithms have $O(N^2)$ running time.} Nodes are removed first that cause the largest cascade of nodes dropping out in the subsequent trimming process. The high degree version of this approach additionally focusses on the nodes of maximum degree, e.g.\ by picking nodes according to $\arg\max_i \left[ d_i+s_i \right]$ with $s_i$ now being the size of the corona, this time not limited to the direct neighborhood, but rather the total graph. Here we merely report that this greedy procedures tend to perform worse than \WN\  when $k\ll d-1$ and becomes comparable when $k \approx d-1$.

Next we contrast the \coreHD\  performance and the \WN\  performance with the performance of recently introduced~\cite{MakseMorone2017} algorithm \CITM-$L$ that uses neighborhood up to distance $L$ and shown in~\cite{MakseMorone2017} to outperform a range of more basic algorithm. 
In figure \ref{fig:compare_er} we compare the performances of the \WN\   algorithm, \coreHD, and \CITM-$10$ (beyond $L=10$ resulted in negligible improvements) on Erd\H{o}s-R\'{e}nyi graphs.
We observe that the \WN\  algorithm outperforms not only the \coreHD\  algorithm, but also the \CITM\  algorithm. We conclude that the optimal use of information about the neighborhood is not given by the \CITM\  algorithm. What the optimal strategy, that only uses local information up to a given distance, remains an intriguing open problem for future work.

To discuss a little more the results observed in Figure \ref{fig:compare_er}, for small $k$ the \coreHD\  algorithm outperforms the \CITM\  algorithm, but when $k$ is increase the performance gap between them shrinks, and for large $k$ (e.g. in Fig. \ref{fig:compare_er} part (a)) CIMT outperforms \coreHD. Both \coreHD\  and \CITM\ are outperformed by the \WN\ algorithm in all the cases we tested. In addition to the results on Erd\H{o}s-R\'{e}nyi graphs in figure \ref{fig:compare_er} we summarize and compare all the three considered algorithms and put them in perspective to the cavity method  results on regular random graphs in table \ref{tbl:compare_algorithms}.

Finally, following up on~\cite{Braunstein2016} we mention that in applications of practical interest it is possible to improve each of the mentioned algorithm by adding an additional random process that attempts to re-insert nodes. Consider the set $\mathcal{S} \subset \mathcal{G}$ of nodes that, when removed, yield an empty $k$-core. Then find back the nodes in $\mathcal{S}$ that can be re-inserted into the graph without causing the $k$-core to re-appear.

\section{Conclusion}

In this paper we study the problem of what is the smallest set of nodes to be removed from a graph so that the resulting graph has an empty $k$-core. The main contribution of this paper is the theoretical analysis of the performance of the \coreHD\  algorithm, proposed originally in~\cite{ZdeborovaZhangZhou2016}, for random graphs from the configuration model with bounded maximum degree. To that end a deterministic description of the associated random process on sparse random graphs is derived that leads to a set of non-linear ordinary differential equations. From the stopping time -- the time at which the $k$-core disappears -- of these differential equations we extract an upper bound on the minimal size of the contagious set of the underlying graph ensemble. The derived upper bounds are considerably better than previously known ones.

Next to the theoretical analysis of \coreHD\, we proposed and investigated numerically several other simple strategies to improve over the \coreHD\  algorithm. All these strategies conserve the essential running time of $O(N)$ on graphs with maximum degree of $O(1)$.
Among our proposals we observe the best to be the \WN\ algorithm. It is based on selecting large degree nodes from the $k$-core that have neighbors of low average degree. In numerical experiments on random regular and Erd\H{o}s-R\'enyi graphs we show that the \WN\ algorithm outperforms \coreHD, as well as other scalable state-of-the-art algorithms~\cite{MakseMorone2017}.

There are several directions that the present paper does not explore and that would be interesting project for future work. One is generalization of the differential equations analysis to the \WN\ algorithm. This should in principle be possible for the price of having to track the number of neigbors of a given type, thus increasing considerably the number of variables in the set of differential equations. Another interesting future direction is optimization of the removal rule using node and its neigborhood up to a distance $L$. It is an open problem if there is a method that outperforms the WN algorithm and uses only information from nearest neighbors. Yet another direction is comparison of the algorithmic performance to message passing algorithms as developed in~\cite{altarelli2013large,altarelli2013optimizing}. Actually the work of~\cite{MakseMorone2017} compares to some version of message passing algorithms and finds that the CIMT algorithm is comparable. Our impression is, however, that along the lines of~\cite{Braunstein2016} where the message passing algorithm was optimized for the dismantling problem, analogous optimization should be possible for the removal of the $k$-core, yielding better results. This is an interesting future project.

\section{Acknowledgement}

We are thankful to the authors of~\cite{MakseMorone2017} for sharing their implementation of the \CITM\ algorithm with us. We would further like to thank Guilhem Semerjian for kind help, comments and indications to related work. LZ acknowledges funding from the European Research Council (ERC) under the European Union’s Horizon 2020 research and innovation programme (grant agreement No 714608 - SMiLe). This work is supported by the ``IDI 2015'' project funded by the IDEX Paris-Saclay, ANR-11-IDEX-0003-02

\bibliography{ref_k_core}

\appendix

\section{Proofs}

\subsection{Proof of Lemma~\ref{lem:stoptime}}
\label{appendix_lem_stoptime_proof}

Based on some worst-case assumptions, we can analyze the evolution of the number of degree-$k_<$ edges.
In particular, in the worst case, the initial removal of a degree-$d$ node can generate $d(k-1)$ degree-$k_<$ edges (i.e., all of its edges were attached to degree-$k$ nodes).
During removal, the number of degree-$k$ edges can also increase by at most $dk$.
Though both of these events cannot happen simultaneously, our worst-case analysis includes both of these effects.

During the $m$-th trimming step, a random edge is chosen uniformly from the set of edges adjacent to nodes of degree $k_<$ (i.e., degree less than $k$).
Let the random variable $Z_m \in \{ k_<,k,k+1,\ldots,d\}$ equal the degree of the node adjacent to the other end and the random variable $X_m$ equal the overall change in the number of degree-$k_<$ edges. 

For $m=1$, this edge is distributed according to
\[ \Pr (Z_1 = z) = \begin{cases} \frac{h_z'}{M-1} & \text{if }z \geq k \\ \frac{h_{k_<}' - 1}{M-1} & \text{if }z = k_<, \end{cases} \]
where $h_z'$ denotes the number of edges of degree $z$ before trimming.
If $Z_1=k_<$, then the edge connects two degree-$k_<$ nodes and removal reduces the number of $k_<$ edges by $X_1=-2$.
If $Z_1=k$, then the edge connects a degree-$k$ node with a degree-$k_<$ node and removal reduces the number of degree-$k$ edges by $k$ and increases the number of degree-$k_<$ edges by $X_1 = k-2$.
If $Z_1>k$, then removal decreases the number of degree-$Z_1$ edges by $Z_1$ and increases the number of degree-$(Z_1-1)$ edges by $Z_1-1$.
In this case, the number of degree-$k_<$ edges is changed by $X_1 = -1$.

The crux of the analysis below is to make additional worst-case assumptions.
One can upper bound the probability of picking degree-$k$ edges because, after $m$ steps of trimming, the number of degree-$k$ edges can be at most $h_k + dk + mk$ (i.e., if initial node removal generates $dk$ edges of degree-$k$ and $Z_m=k+1$ during each step of trimming).
Also, one can upper bound the number of degree-$k_<$ edges because at least one is removed during each step.
To upper bound the random variable $T$, we use the following worst-case distribution for $X_m$,
\begin{equation}
\Pr (X_m = x) = \begin{cases} \frac{h_k+dk+mk}{M-2d-2m} & \text{if }x = k-2 \\ \frac{M-2d-h_k-dk-mk}{M-2d-2m} & \text{if }x = -1. \end{cases} \label{eq:worst_case_x}
\end{equation}
In this formula, the $2d$ term represents initial removal of a degree-$d$ node and the $2m$ term represents the edge removal associated with $m$ steps of trimming.
We note that, since edges attach two nodes of different degrees, all edges are counted twice in $\bm{h}$.
Similarly, the $dk$ term represents the worst-case increase in the number of degree-$k$ edges during the initial removal.

Now, we can upper bound the number of degree-$k_<$ edges after $m$ steps of trimming
by the random sum \[S_{m}=d(k-1)+X_{1}+X_{2}+\cdots+X_{m},\] where each $X_i$ is drawn independently according to the worst-case distribution~\eqref{eq:worst_case_x}.
The term $d(k-1)$ represents the worst-case event that the initial node removal generates $d(k-1)$ edges of degree-$k_<$.
By choosing $M_0 (\delta)$ appropriately, our initial assumption about $\bm{h}$ implies that $E[X_{i}]\leq-2\delta$ for all $M>M_0 (\delta)$.
Thus, the random variable $S_m$ will eventually become zero with probability 1.
Since $S_m$ upper bounds the number of degree-$k_<$ edges after $m$ steps of trimming, when it becomes zero (say at time $m'$), it follows that the stopping time satisfies $T\leq m'$.

To complete the proof, we focus on the case of $m' = (k-1)^2 \delta^{-2} \ln M$.
Since the distribution of $X_m$ gradually places more weight on the event $X_m=k-2$ as $m$ increases, we can also upper bound $S_{m'}$ by replacing $X_1,\ldots,X_{m'-1}$ by i.i.d.\ copies of $X_{m'}$.
From now on, let $Y$ denote a random variable with the same distribution as $X_{m'}$ and define $S' =  Y_1+Y_2+\cdots+Y_{m'}$ to be a sum of $m'$ i.i.d.\ copies of $Y$.
Then, we can write
\[ \Pr (T \leq m') \geq \Pr (S' \leq -d(k-1)) = 1 - \Pr (S' > -d(k-1)). \]
Next, we observe that~\eqref{eq:worst_case_x} converges in distribution (as $M\to \infty$) to a two-valued random variable $Y^*$ that places probability at most $\frac{1-3\delta}{k-1}$ on the point $k-2$ and the remaining probability on $-1$.
Since $E[Y^*]=-3\delta$, there is an $M_0 (\delta)$ such that $E[Y]=-2\delta$ for all $M>M(\delta)$
Finally, Hoeffding's inequality implies that
\[ \Pr (S' > -\delta m') \leq e^{-2\delta^2 m' / (k-1)^2} = M^{-2}. \]
Putting these together, we see that
\[ \Pr (T > m') \leq \Pr (S' > -d(k-1)) \leq M^{-2} \]
for $m' > d(k-1)/\delta$.
As $m'=(k-1)^2 \delta^{-2} \ln M$, the last condition can be satisfied by choosing $M(\delta) > e^{\delta^2 d/(k-1)}$.
Since this condition is trivial as $\delta \to 0$, the previous convergence condition will determine the minimum value for $M_0 (\delta)$ when $\delta$ is sufficiently small.

\subsection{Proof Sketch for Lemma~\ref{lem:diffeq}}
\label{appendix_lem_diffeq_proof}

%\begin{proof} [Sketch of Proof]
We note that graph processes like \coreHD\  often satisfy the conditions necessary for fluid limits of this type.
The only challenge for the \coreHD\  process is that the $\bm{h}(n)$ process may have unbounded jumps due to the trimming operation.
To handle this, we use Lemma~\ref{lem:stoptime} to show that, with high probability, the trimming process terminates quickly when $\bm{h}(n)$ satisfies the stated conditions.
Then, the fluid limit of the \coreHD\  process follows from~\cite[Theorem~5.1]{Wormald-lara99}.

To apply~\cite[Theorem~5.1]{Wormald-lara99}, we first fix some $\delta>0$ and choose the valid set $D$ to contain normalized degree distributions satisfying $\eta_k (t) \leq \frac{1-4\delta}{k-1} \left| \boldsymbol\eta (t) \right|$ and $\left|\boldsymbol{\eta}(t)\right|\geq \delta$.
This choice ensures that Lemma~\ref{lem:stoptime} can be applied uniformly for all $\boldsymbol{\eta}(t) \in D$.
Now, we describe how the constants are chosen to satisfy the necessary conditions of the theorem and achieve the stated result.
For condition $(i)$ of~\cite[Theorem~5.1]{Wormald-lara99}, we will apply Lemma~\ref{lem:stoptime} but first recall that the initial number of nodes in the graph is denoted by $N$ and hence $\left|\boldsymbol{\eta}(t)\right|\geq \delta$ implies $M=|\bm{h}(n)| \geq B N$ with high probability for some $B>0$.
Thus,  condition $(i)$ can be satisfied by choosing\footnote{The additional factor of $d$ is required because each trimming step can change the number of degree-$j$ edges by at most $d$.} $\beta(N) = (k-1)^2 \delta^{-2} d \ln N$ applying Lemma~\ref{lem:stoptime} with $M=BN$ to see that $\gamma(N) = O(N^{-2})$.
To verify condition $(ii)$ of the theorem, we note that~\eqref{eq:non_linear_dgl} is derived from the large-system limit of the drift and the error $\lambda(n)$ can be shown to be $O(N^{-1})$. 
To verify condition $(iii)$ of the theorem, we note that that~\eqref{eq:non_linear_dgl} is Lipschitz on $D$.
Finally, we choose $\lambda(N) = N^{-1/4} \geq \lambda_1 (N) + C_0 N \gamma(N)$ for large enough $N$.
Since
\[ \beta(N)/\lambda(N) e^{-N\lambda(N)^3 / \beta(N)^3} = O\left( \frac{\ln N}{N^{-1/4}} e^{-A N^{1/4} / (\ln N)^3}\right) = O \left( N^{-1} \right), \]
these choices imply that the \coreHD\  process concentrates around $\boldsymbol{\eta}(t)$ as stated.
We only sketch this proof because very similar arguments have used previously for other graph processes~\cite{Achlioptas-stoc00,Achlioptas-tcs01,MonassonCocco01}.

\end{document}